\documentclass{eptcs}

\usepackage[utf8]{inputenc}
\usepackage{amsfonts}
\usepackage{amssymb,amsmath}
\usepackage{mathrsfs}
\usepackage{pictex}
\usepackage{theorem}

{\theorembodyfont{\slshape}
\newtheorem{theorem}{Theorem}[section]
\newtheorem{proposition}[theorem]{Proposition}
\newtheorem{corollary}[theorem]{Corollary}

{\theorembodyfont{\rmfamily}
\newtheorem{definition}[theorem]{Definition}
}

\def\eproof{{\mbox{}\hfill\eproof}\medskip}

\begin{document}

\makeatletter

\def\tv{{\; \wr \;}}

\newcommand{\real}{{\mathbb{R}}}
\newcommand{\R}{{\mathbb{R}}}
\newcommand{\C}{{\mathbb{C}}}
\newcommand{\Q}{{\mathbb{Q}}}
\newcommand{\Z}{{\mathbb{Z}}}
\newcommand{\N}{{\mathbb{N}}}
\newcommand{\F}{{\mathbb{F}}}
\newcommand{\K}{{\mathbb{K}}}
\newcommand\mK{\mathcal{K}}
\newcommand\CC{\mathcal{C}}
\newcommand\DD{\mathcal{D}}
\newcommand\mL{\mathcal{L}}
\newcommand\mT{\mathcal{T}}
\newcommand\mF{\mathcal{F}}
\newcommand\mH{\mathcal{H}}
\newcommand\mN{\mathcal{N}}
\newcommand\cP{\mathcal{P}}
\newcommand\mC{{\mathscr C}}
\newcommand\Oh{{\cal O}}
\newcommand\comment[1]{\footnote{#1}}
\def\scrS{{\mathscr S}}
\def\alt{{\mathsf{alt}}}
\def\Poly{{\mathsf{Poly}}}
\def\Exp{{\mathsf{Exp}}}
\def\Log{{\mathsf{Log}}}
\def\Const{{\mathsf{Const}}}
\def\Alt{{\mathsf{Alt}}}
\def\Maj{{\mathsf{Maj}}}
\def\bfy{{\mathbf y}}


\def\BP{{\rm BP}}
\def\P{{\sf P}}
\def\FP{{\sf FP}}
\def\NP{{\sf NP}}
\def\coNP{{\sf coNP}}
\def\NC{{\sf NC}}
\def\FNC{{\sf FNC}}
\def\PSPACE{{\sf PSPACE}}
\def\NPSPACE{{\sf NPSPACE}}
\def\FPSPACE{{\sf FPSPACE}}
\def\PH{{\sf PH}}
\def\cSigma{{\mathsf \Sigma}}
\def\cPi{{\mathsf \Pi}}
\def\MAE{{\rm MA}{\mathsf \exists}_{\{ { \bf 0},{ \bf 1} \}ern-0.4pt\R}}
\def\MAV{{\rm MA}\{ { \bf 0},{ \bf 1} \}ern-0.5pt{\mathsf \forall}_{\{ { \bf 0},{ \bf 1} \}ern-0.4pt\R}}
\def\PHR{{\rm PH}_{\R}}
\def\NCR{{\rm NC}_{\R}}
\def\QHR{{\rm QH}_{\R}}
\def\EXPR{{\rm EXP}_{\R}}
\def\PAREXP{{\rm PAREXP}_{\R}}
\def\PATR{{\rm PAT}_{\R}}
\def\DPATR{{\rm DPAT}_{\R}}
\def\EXP{{\sf EXP}}
\def\PAR{{\rm PAR}}
\def\NPAR{{\rm NPAR}}
\def\coNPAR{{\rm coNPAR}}
\def\DNPAR{{\rm DNPAR}}
\def\VP{{\rm VP}}
\def\VNP{{\rm VNP}}
\def\coRP{{\sf coRP}}
\def\CNP{{\#\NP}}
\def\poly{\mbox{\small\sf poly}}
\def\CP{{\#\P}}
\def\Dig{{\rm D}}
\def\PR{{\rm P}_{\R}}
\def\NPR{{\rm NP}_{\R}}
\def\coNPR{{\rm coNP}_{\R}}
\def\co{{\rm co}}
\def\msL{\mathscr L}

\title{Pointers in Recursion: Exploring the Tropics \thanks{Work partially supported by ANR project ELICA -  ANR-14-CE25-0005}}

\author{Paulin Jacob\'e de Naurois
\institute{CNRS, Universit\'e Paris 13, Sorbonne Paris Cit\'e, LIPN,  UMR 7030, F-93430 Villetaneuse, France.} 
\email{\tt denaurois@lipn.univ-paris13.fr}}
\def\titlerunning{Pointers in Recursion: Exploring the Tropics}
\def\authorrunning{Paulin Jacob\'e de Naurois}

\maketitle
\thispagestyle{empty}

\begin{quote}
{\small {\bf Abstract.}
We translate the usual class of partial/primitive recursive functions to a pointer recursion framework, accessing actual input values via a pointer reading unit-cost function. These pointer recursive functions classes are proven equivalent to the usual partial/primitive recursive functions. Complexity-wise, this framework captures in a streamlined way most of the relevant sub-polynomial classes. Pointer recursion with the safe/normal tiering discipline of Bellantoni and Cook corresponds to polylogtime computation. We introduce a new, non-size increasing tiering discipline, called tropical tiering. Tropical tiering and pointer recursion, used with some of the most common recursion schemes, capture the classes logspace, logspace/polylogtime, ptime, and NC. Finally, in  a fashion reminiscent of the safe recursive functions, tropical tiering is expressed directly in the syntax of the function algebras, yielding the tropical recursive function algebras.

}
\end{quote}

\section*{Introduction}

Characterizing complexity classes without explicit reference to the computational model used for defining these classes, and without explicit bounds on the resources allowed for the calculus, has been a long term goal of several lines of research in computer science. One rather successful such line of research is recursion theory. The foundational work here is the result of Cobham~\cite{Cob:65}, who gave a characterization of polynomial time computable functions in terms of bounded recursion on notations - where, however, an explicit polynomial bound is used in the recursion scheme. Later on, Leivant~\cite{DBLP:conf/lics/Leivant91} refined this approach with the notion of tiered recursion: explicit bounds are no longer needed in his recursion schemes. Instead, function arguments are annotated with a static, numeric denotation, a \emph{tier},  and a tiering discipline is imposed upon the recursion scheme to enforce a polynomial time computation bound. A third important step in this line of research is the work of Bellantoni and Cook~\cite{DBLP:journals/cc/BellantoniC92}, whose safe recursion scheme uses only syntactical constraints akin to the use of only two tier values, to characterize, again, the class of polynomial time functions. 

Cobham's approach has also later on been fruitfully extended to other, important complexity classes. Results relevant to our present work, using explicitly  bounded recursion, are those of Lind~\cite{TR:Lind} for logarithmic space, and Allen~\cite{DBLP:journals/apal/Allen91} and Clote~\cite{Clote:89} for small parallel classes. 

Later on, Bellantoni and Cook's purely syntactical approach proved also useful for characterizing other complexity classes. Leivant and Marion~\cite{leivant:inria-00099077,DBLP:journals/tcs/LeivantM00} used a predicative version of the safe recursion scheme to characterize alternating complexity classes, while Bloch~\cite{DBLP:journals/cc/Bloch94}, Bonfante et al~\cite{DBLP:journals/iandc/BonfanteKMO16} and Kuroda\cite{DBLP:journals/cc/Kuroda04}, gave characterizations of small, polylogtime, parallel complexity classes. An important feature of these results is that they use, either explicitly or not, a tree-recursion on the input. This tree-recursion is implicitly obtained in Bloch's work by the use of an extended set of basic functions, allowing for a dichotomy recursion on the input string, while it is made explicit in the recursion scheme in the two latter works. As a consequence, these characterizations all rely on the use of non-trivial basic functions, and non-trivial data structures. Moreover, the use of distinct basic function sets and data structures make it harder to express these charcterizations in a uniform framework. 

Among all these previous works on sub-polynomial complexity classes, an identification is assumed between the argument of the functions of the algebra, on one hand, and the computation input on the other hand: an alternating, logspace computation on input $\overline{x}$ is denoted by a recursive function with argument $\overline{x}$. While this seems very natural for complexity classes above linear time, it actually yields a fair amount of technical subtleties and difficulties for sub-linear complexity classes. Indeed, following Chandra et al.~\cite{DBLP:journals/jacm/ChandraKS81} seminal paper, sub-polynomial complexity classes need to be defined with a proper, subtler model than the one-tape Turing machine: the random access Turing machine (RATM), where computation input is accessed via a unit-cost pointer reading instruction. RATM input is thus accessed via a read-only instruction, and left untouched during the computation - a feature quite different to that of a recursive function argument.  Our proposal here is to use a similar construct for reading the input in the setting of recursive functions: our functions will take as input pointers on the computation input, and one-bit pointer reading will be assumed to have unit cost. Actual computation input are thus implicit in our function algebras: the fuel of the computational machinery is only pointer arithmetics. This proposal takes inspiration partially from the Rational Bitwise Equations of ~\cite{DBLP:journals/iandc/BonfanteKMO16}.

Following this basic idea, we then  introduce a new tiering discipline, called \emph{tropical tiering}, to enforce a non-size increasing behavior on our recursive functions, with some inspirations taken from previous works of M. Hofmann~\cite{DBLP:journals/iandc/Hofmann03,DBLP:journals/tocl/HofmannS10}. Tropical tiering induces a polynomial interpretation in the tropical ring of polynomials (hence its name), and yields a characterization of logarithmic space. The use of different, classical recursion schemes yield characterizations of other, sub-polynomial complexity classes such as polylogtime,  NC, and the full polynomial time class. Following the approach of Bellantoni and Cook, we furthermore embed the tiering discipline directly in the syntax, with only finitely many different tier values - four tier values in our case, instead of only two tier values for the safe recursive functions, and provide purely syntactical characterizations of these complexity classes in a unified, simple framework. Compared to previous works, our framework uses a unique, and rather minimal set of unit-cost basic functions, computing indeed basic tasks, and a unique and also simple data structure. 


The paper is organized as follows. Section~\ref{sec:defs} introduces the notations, and the framework of pointer recursion. Section~\ref{sec:pr} applies this framework to primitive recursion. Pointer partial/primitive recursive functions are proven to coincide with their classical counterparts in Theorem~\ref{theo:pr}. Section~\ref{sec:safe} applies this framework to safe recursion on notations. Pointer safe recursive functions are proven to coincide with polylogtime computable functions in Theorem~\ref{theo:safe}. Tropical tiering is defined in Section~\ref{sec:tropics}. Proposition~\ref{prop:trop-inter} establishes the tropical interpretation induced by tropical tiering. Tropical recursive functions are then introduced in Subsection~\ref{subsec:trop}. 
Section~\ref{sec:LP} gives a sub-algebra of the former, capturing logspace/polylogtime computable functions in Theorem~\ref{thm:logspacepolylogtime}. Finally, Section~\ref{sec:alt} explores tropical recursion with substitutions, and provides a characterization of P in Theorem~\ref{thm:P} and of NC in Theorem~\ref{thm:NC}.

\section{Recursion}\label{sec:defs}

\subsection{Notations, and Recursion on Notations}

Data structures considered in our paper are finite words over a finite alphabet. For the sake of simplicity, we consider the finite, boolean alphabet $\{{\bf 0}, {\bf 1}\}$. The set of finite words over $\{{\bf 0}, {\bf 1}\}$ is denoted as $\{{\bf 0}, {\bf 1}\}^*$. 

Finite words over  $\{{\bf 0}, {\bf 1}\}$ are denoted with overlined variables names, as in $\overline{x}$. Single values in $\{{\bf 0}, {\bf 1}\}$ are denoted as plain variables names, as in $x$. The empty word is denoted by $\varepsilon$, while the dot symbol "." denotes the concatenation of two words as in $a . \overline{x}$, the finite word obtained by adding an $a$ in front of the word $\overline{x}$. Finally, finite arrays of boolean words are denoted with bold variable names, as in ${\bf x} = (\overline{x_1}, \cdots, \overline{x_n})$. When defining schemes, we will often omit the length of the arrays at hand, when clear from context, and use bold variable names to simplify notations. Similarly, for mutual recursion schemes, finite arrays of mutually recursive functions are denoted by a single bold function name. In this case, the \emph{width} of this function name is the size of the array of the mutually recursive functions. 

Natural numbers are identified with finite words over $\{ {\bf 0}, {\bf 1} \}$ via the usual binary encoding. Yet, in most of our function algebras, recursion is not performed on the numerical value of an integer, as in classical primitive recursion, but rather on  its boolean encoding, that is, on the finite word over $\{ {\bf 0}, {\bf 1} \}$ identified with it: this approach is denoted as \emph{recursion on notations}.  

\subsection{Turing Machines with Random Access}

When considering sub-polynomial complexity class, classical Turing Machines often fail to provide a suitable cost model. A crucial example is the class DLOGTIME: in logarithmic time, a classical Turing machine fails to read any further than the first $k.\log(n)$ input bits. In order to provide a suitable time complexity measure for sub-polynomial complexity classes, Chandra et al~\cite{DBLP:journals/jacm/ChandraKS81} introduced the Turing Machine with Random Access (RATM), whose definition follows.

\begin{definition}{RATM}\\
A Turing Machine with Random Access (RATM) is a Turing machine with no input head, one (or several) working tapes and a special \emph{pointer}  tape, of logarithmic size, over a binary alphabet. The Machine has a special \emph{Read} state such that, when the binary number on the pointer tape is $k$, the transition from the Read state consists in writing the $k^{th}$ input symbol on the (first) working tape. 

\end{definition}

\subsection{Recursion on Pointers}

In usual recursion theory, a function computes a value on its input, which is given explicitly as an argument. This, again, is the case in classical primitive recursion. While this is suitable for describing explicit computation on the input, as, for instance for single tape Turing Machines, this is not so for describing input-read-only computation models, as, for instance,  RATMs. In order to propose a suitable recursion framework for input-read-only computation, we propose the following \emph{pointer recursion} scheme, whose underlying idea is pretty similar to that of the RATM. 

As above, recursion data is given by finite, binary words, and the usual recursion on notation techniques on these recursion data apply. The difference lies in the way the actual computation input is accessed: in our framework, we distinguish two notions, the \emph{computation input}, and the \emph{function input}: the former denotes the input of the RATM, while the latter denotes the input in the function algebra. For classical primitive recursive functions, the two coincide, up to the encoding of integer into binary strings. In our case, we assume an explicit encoding of the former into the latter, given by the two following constructs.


Let $\overline{w} = w_1. \cdots. w_n \in \{ {\bf 0}, {\bf 1} \} ^*$ be a computation input. To $\overline{w}$,  we associate two constructs,
\begin{itemize}
\item the {\tt Offset}: a finite word over $\{ {\bf 0}, {\bf 1} \}$, encoding in binary the length $n$ of $\overline{w}$, and
\item the {\tt Read} construct, a 1-ary function, such that, for any binary encoding $\overline{i}$ of an integer $0 < i \le n$, ${\tt Read}(\overline{i}) = w_i$, and, for any other value $\overline{v}$, ${\tt Read}(\overline{v}) = \varepsilon$.
\end{itemize}

Then, for a given \emph{computation input} $\overline{w}$, we fix accordingly the semantics of the {\tt Read} and {\tt Offset}  constructs as above, and a \emph{Pointer Recursive function} over $\overline{w}$ is evaluated with sole  input the {\tt Offset}, accessing computation input bits via the {\tt Read} construct. For instance, under these conventions, ${\tt Read}({\tt hd} ({\tt Offset}))$ outputs the first bit of the computational input $\overline{w}$. In some sense, the two constructs depend on $\overline{w}$, and can be understood as functions on $\overline{w}$. However, in our approach, it is important to forbid $\overline{w}$ from appearing explicitly as a function argument  in the syntax of the function algebras we will define, and from playing any role in the composition and recursion schemes. Since $\overline{w}$ plays no role at the syntactical level - its only role is at the semantical level-  we chose to remove it completely $\overline{w}$ from the syntactical definition of our functions algebras.

%


\section{Pointers Primitive Recursion}\label{sec:pr}

Let us first detail our pointer recursive framework for the classical case of primitive recursion on notations.

\begin{paragraph}{Basic pointer functions.}
Basic pointer functions are the following kind of functions:

\begin{enumerate}
\item Functions manipulating finite words over $\{ { \bf 0},{ \bf 1} \}$. For any $a \in \{ { \bf 0},{ \bf 1} \}, \overline{x} \in \{ { \bf 0},{ \bf 1} \}^*$, 
$$
\begin{array}{rclrclrcl}
{\tt hd}(a.\overline{x}) & = & a &\;\;\; {\tt tl}(a.\overline{x}) &  = & \overline{x} &\;\;\; {\tt s}_0(\overline{x}) &=& 0.\overline{x}\\
{\tt hd}(\varepsilon) & = & \varepsilon & {\tt tl}(\varepsilon) &  = & \varepsilon & {\tt s}_1(\overline{x}) &=& 1.\overline{x}\\
\end{array}$$
\item Projections. For any $n \in \N$, $1 \le i \le n$, 
$${\tt Pr}_i^n(\overline{x_1}, \cdots, \overline{x_n}) = \overline{x_i}$$ 
\item and, finally, the {\tt Offset } and {\tt Read} constructs, as defined above.
\end{enumerate}
\end{paragraph}
\begin{paragraph}{Composition.}

Given functions $g$, and $h_1, \cdots, h_n$, we define $f$ by composition as 
$$f({\bf x}) = g(h_1({\bf x}), \cdots, h_n({\bf x})).$$
\end{paragraph}
\begin{paragraph}{Primitive Recursion on Notations.} Let $\perp$ denote non-terminating computation.
 Given functions $h$, $g_0$ and $g_1$, we define $f$ by primitive recursion on notations as 
 \begin{eqnarray*}
f(\varepsilon, {\bf y}) &=& h({\bf y})\\
f({\tt s}_a(\overline{x}),  {\bf y}) &=& \left \{ \begin{array}{lr} g_a(\overline{x}, f(\overline{x},  {\bf y}), {\bf y}) & \mbox{ if } f(\overline{x},  {\bf y}) \neq \perp \\ \perp & \mbox{ otherwise.} \end{array} \right.
\end{eqnarray*}

\end{paragraph}
\begin{paragraph}{Minimization.}
For a function $s$, denote by $s^{(n)}$ its $n^{th}$ iterate. Then, given  a function $h$, we define $f$ by minimization on $\overline{x}$ as  
$$ \mu \overline{x}( h( \overline{x}, {\bf y})) = \left\{ \begin{array}{lr} \perp & \mbox{ if } \forall t \in \N, {\tt hd}(h( s_0^{(t)}(\varepsilon),{\bf y})) \neq {\tt s}_1(\varepsilon) \\ s_0^{(k)}(\varepsilon) & \mbox{where } k = \min\{t \; : {\tt hd}(h( s_0^{(t)}(\varepsilon),{\bf y})) = {\tt s}_1(\varepsilon) \} \;\; \mbox{ otherwise.} \end{array} \right.$$
In other words, a function $f$ defined by  minimization on $h$  produces the shortest sequence of $0$ symbols satisfying a simple condition on $h$ , if it exists.
\end{paragraph}

Let now $PR_{not}^{point}$ be the closure of basic pointer functions under composition and primitive recursion on notations, and $REC_{not}^{point}$ be the closure of basic pointer functions under composition, primitive recursion on notations, and minimization. Then, as expected,

\begin{theorem}\label{theo:pr}
Modulo the binary encoding of natural integers, $PR_{not}^{point}$ is the classical class of primitive recursive functions, and  $REC_{not}^{point}$ is the classical class of recursive functions.
\end{theorem}
\begin{proof}
It is already well known that primitive recursive functions on notations are the classical primitive recursive functions, and recursive functions on notations are the classical recursive functions. Now, for one direction, it suffices to express the {\tt Read} and {\tt Offset} basic pointer functions as primitive recursive functions on the computation input. For the other direction, it suffices to reconstruct with pointer primitive recursion the computation input from the {\tt Read} and {\tt Offset} basic pointer functions. 
\end{proof}

\section{Pointer Safe Recursion}\label{sec:safe}

We recall the tiering discipline of Bellantoni and Cook~\cite{DBLP:journals/cc/BellantoniC92}: functions arguments are divided into two tiers, \emph{normal} arguments and \emph{safe} arguments. Notation-wise, both tiers are separated by a semicolon symbol in a block of arguments, the normal arguments being on the left, and the safe arguments on the right. We simply apply this tiering discipline to our pointer recursion framework. 

\begin{paragraph}{Basic Pointer Safe Functions.}

Basic pointer safe functions are the basic pointer functions of the previous section, all their arguments being considered safe.
\end{paragraph}
\begin{paragraph}{Safe Composition.}

Safe composition is somewhat similar to the previous composition scheme, with a tiering discipline, ensuring that safe arguments cannot be moved to a normal position in a function call. The reverse however is allowed. 
$$f({\bf x};{\bf y}) = g(h_1({\bf x};) , \cdots, h_m({\bf x};);  h_{m+1}({\bf x};{\bf y}), \cdots , h_{m+n}({\bf x};{\bf y})).$$

Calls to functions $h_{m+i}$, where safe arguments are used, are placed in safe position in the argument block of $g$. A special case of safe composition is 
$f(\overline{x};\overline{y}) = g(;\overline{x},\overline{y})$, where a normal argument $\overline{x}$ is used in safe position in a call. Hence, we liberally use normal arguments in safe position, when necessary. 
\end{paragraph}
\begin{paragraph}{Safe Recursion.}
The recursion argument is normal. The recursive call is placed in safe position, a feature that prevents nesting recursive calls exponentially. 
\begin{eqnarray*}
f(\varepsilon, {\bf y}; {\bf z}) &=& h({\bf y}; {\bf z})\\
f(a.\overline{x}, {\bf y}; {\bf z}) &=& g_a(\overline{x}, {\bf y}; f(\overline{x}, {\bf y}; {\bf z}), {\bf z}) .
\end{eqnarray*}

\end{paragraph}

Let now $SR_{not}^{point}$ be  the closure of the basic pointer safe functions under safe composition and safe recursion. 
\begin{theorem}\label{theo:safe}
$SR_{not}^{point}$ is the class $DTIME(polylog)$ of functions computable in poly-logarithmic time. 
\end{theorem}

\begin{proof}
The proof is essentially the same as for the classical result by Bellantoni and Cook~\cite{DBLP:journals/cc/BellantoniC92}. Here however, it is crucial to use the RATM as computation model. Simulating a polylogtime RATM with safe recursion on pointers is very similar to simulating a polytime TM with safe recursion - instead of explicitly using the machine input as recursion data, we use the size of the input as recursion data, and access the input values via the {\tt Read} construct, exactly as is done by the RATM model.  The other direction is also similar: the tiering discipline of the safe recursion on pointers enforces a polylog bound on the size of the strings (since the initial recursion data - the {\tt Offset} - has size logarithmic in the size $n$ of the computation input), and thus a polylog bound on the computation time. 
\end{proof}

\section{Tropical Tiering}\label{sec:tropics}

We present here another, stricter tiering discipline, that we call \emph{tropical Tiering}. The adjective "tropical" refers to the fact that this  tiering induces a polynomial interpretation in the tropical ring of polynomials. This tiering discipline takes some inspiration from Hofmann's work on non-size increasing types~\cite{DBLP:journals/iandc/Hofmann03}, and pure pointer programs~\cite{DBLP:journals/tocl/HofmannS10}. The idea however  is to use here different tools than  Hofmann's  to achieve a similar goal of bounding the size of the function outputs. We provide here a non-size increasing discipline via the use of tiering, and  use it in the setting of pointer recursion to capture not only pure pointer programs (Hoffman's class), but rather pointer programs with pointer arithmetics, which is in essence the whole class Logspace.

\begin{paragraph}{Basic Pointer Functions.}

We add the following numerical successor basic function. Denote by $E: \N \rightarrow \{ {\bf 0},{\bf 1} \}^*$ the usual binary encoding of integers, and $D: \{ {\bf 0},{\bf 1} \}^* \rightarrow \N$ the decoding of binary strings to integers. Then, 
$${\tt s}(\overline{x}) = E(D(\overline{x})+1)$$

denotes the numerical successor on binary encodings, and, by convention, $\varepsilon$ is the binary encoding of the integer $0$.

\end{paragraph}
\begin{paragraph}{Primitive Recursion on Values.}

Primitive recursion on values is the usual primitive recursion, encoded into binary strings:

\begin{eqnarray*}
f(\varepsilon, {\bf y}) &=& h({\bf y})\\
f({\tt s}(\overline{x}),  {\bf y}) &=&  g(\overline{x}, f(\overline{x},  {\bf y}), {\bf y}).
\end{eqnarray*}
\end{paragraph}

\subsection{Tropical Tier}

As usual, tiering consists in assigning function variables to different classes, called tiers. In our setting, these tiers are identified by a numerical value, called \emph{tropical tier}, or, shortly, \emph{tropic}. 
The purpose of our tropical tiers is to enforce a strict control on the increase of the size of the binary strings during computation. Tropics take values in $\Z \cup \{ -\infty \}$. The tropic of the $i^{th}$ variable of a function $f$ is denoted $T_i(f)$. The intended meaning of the tropics is to provide an upper bound on the linear growth of the function output size with respect to the corresponding input size, as per Proposition~\ref{prop:trop-inter}.
Tropics are inductively defined as follows.

\begin{enumerate}
\item Basic pointer functions: 
$$
\begin{array}{rclrclrcl}
T_{j \neq i}({\tt Pr}_i^n) &=& - \infty &\;\;\; T_1({\tt hd}) &=& -\infty &\;\;\; T_1({\tt Read}) &=& -\infty\\
T_1({\tt tl}) &=& -1 &&&&&&\\
T_{ i}({\tt Pr}_i^n) &=& 0 &&&&&&\\
T_1({\tt s}_0) &=& 1 & T_1({\tt s}_1) &=& 1 & T_1({\tt s}) &=& 1
\end{array}
$$

\item Composition: $$T_t(f) = \max_i \{ T_i (g) + T_t(h_i) \}.$$

\item Primitive recursion on notations. 
Two cases arise:
\begin{itemize}
\item 
$T_2(g_0)\le 0$ and $T_2(g_1) \le 0$. In that case,  we set  
\begin{enumerate}
\item $T_1(f) = \max \{T_1(g_0), T_1(g_1), T_2(g_0), T_2(g_1)\}$,  and,
\item for all $t\ge 1$,\\ $T_t(f) = \max \{T_{t+1}(g_0),T_{t+1}(g_1), T_{t-1}(h), T_2(g_0), T_2(g_1)\}$.
\end{enumerate}
\item  the previous case above does not hold, $T_2(g_0)\le 1$, and  $T_2(g_1) \le 1$. In that case, we also require that  $T_1(g_0) \le 0$,  $T_1(g_1) \le 0$, and, for all $t \ge 2$, $T_t(g_0)=T_t(g_1)=T_{t-2}(h) = -\infty$. Then,  we set  $T_1(f) = \max \{T _1(g_0), T_1(g_1), T_2(g_0)-1, T_2(g_1)-1, c_h \}$, where $c_h$ is a constant for $h$ given in Proposition~\ref{prop:trop-inter} below, and, for $t\ge 1$, $T_t(f)=-\infty$.
\end{itemize}
Other cases than the two above do not enjoy tropical tiering.

\item Primitive recursion on values. 
Only one case arises:
\begin{itemize}
\item  $T_2(g)\le 0$ .  In that case,  we set  
\begin{enumerate}
\item$T_1(f) = \max \{T_1(g), T_2(g) \}$,  and,
\item for all $t\ge 1$, $T_t(f) = \max \{T_{t+1}(g), T_{t-1}(h), T_2(g)\}$.
\end{enumerate}
\end{itemize}
Again, other cases than the one above do not enjoy tropical tiering.
\end{enumerate}

Furthermore, when using tropical tiering, we use mutual recursion schemes. For ${\bf f} = (f_1,\cdots,f_n)$, mutual primitive recursion (on values) is classically defined as follows,
\begin{eqnarray*}
{\bf f}(\varepsilon, {\bf y}) &=& {\bf h}({\bf y})\\
{\bf f}({\tt s}(\overline{x}),  {\bf y}) &=& \left \{ \begin{array}{lr} {\bf g}(\overline{x}, {\bf f}(\overline{x},   {\bf y}), {\bf y}) & \mbox{ if } \forall i\; (f_i(\overline{x},  {\bf y}) \neq \perp)\\ \perp & \mbox{ otherwise.} \end{array} \right.
\end{eqnarray*}
and similarly for mutual primitive recursion on notations. Tropical tiering is then extended to mutual primitive recursion in a straightforward manner. 

We define the set of L-primitive pointer recursive functions as the closure of the basic pointer functions of Sections 2 and 4 under composition, (mutual) primitive recursion on notations and (mutual) primitive recursion on values, with tropical tiering.

\subsection{Tropical Interpretation}

Tropical tiering induces a non-size increasing discipline. More formally,

\begin{proposition}\label{prop:trop-inter}
The tropical tiering of a L-primitive recursive function $f$ induces a polynomial interpretation of $f$ on the tropical ring of polynomials, as follows.

For any L-primitve recursive function $f$ with $n$ arguments, there exists a constant $c_f \ge 0$ such that
$$| f(\overline{x_1},\cdots,\overline{x_n}) | \le \max_t \{ T_t(f) + |\overline{x_t}|, c_f \}.$$
\end{proposition}

\begin{proof}
The proof is given for non-mutual recursion schemes, by  induction on the definition tree. Mutual recursion schemes follow the same pattern.
\begin{enumerate}
\item For basic pointer functions, the result holds immediately. 
\item Let $f$ be defined by composition, and assume that the result holds for  the functions $g$, $h_1,\cdots, h_n$ . Then, for any $i=1, \cdots, n$,  $| h_i({\bf x}) | \le \max_t \{ T_t(h_i) + |\overline{x_t}|, c_{h_i} \}$. Moreover, there exists by induction $c_g$ such that  $| g(h_1({\bf x}), \cdots, h_n({\bf x}))  | \le \max_i \{ T_i(g) + | h_i({\bf x}) |, c_g \}$. Composing the inequalities above yields  $| g(h_1({\bf x}), \cdots, h_n({\bf x}))  | \le \max_i \{ T_i(g) + \max_t \{  T_t(h_i)  + |\overline{x_i}|, c_{h_i} \}, c_g\} =  \max_t \{ T_t(f) + |\overline{x_t}|, \max_i\{ c_{f_i}, c_g \} \}$.

\item Let $f$ be defined by primitive recursion on notations, and assume that the first case holds. Let $f(a.\overline{x}, {\bf y}) = g_a(\overline{x}, f(\overline{x}, {\bf y}),  {\bf y})$, for $a\in\{0,1\}$, and assume $T_2(g_0) \le 0$  and $T_2(g_1) \le 0$.
We apply the tropical interpretation on $g$, and we show by induction the result for $f$ on the length of $a.\overline{x}$. 
\begin{enumerate}
\item If $\max_{\overline{x}, f(\overline{x},{\bf y}), t}  \{ |\overline{x}| + T_1(g_a), |f(\overline{x},{\bf y})| +T_2(g_a), |\overline{y_t}| + T_{t+2}(g_a), c_{g_a} \} = |\overline{x}|+ T_1(g_a)$:  $| f(a.\overline{x}, {\bf y})| \le |\overline{x}| + T_1(g_a) \le  |\overline{x}| +T_1(f)$, and the result holds.
\item If $\max_{\overline{x}, f(\overline{x},{\bf y}), t}  \{ |\overline{x}| + T_1(g_a), |f(\overline{x},{\bf y})| +T_2(g_a), |\overline{y_t}| + T_{t+2}(g_a), c_{g_a} \}  = |f(\overline{x},{\bf y})|+T_2(g_a)$: Since $T_2(g_a)\le0$,  $| f(a.\overline{x}, {\bf y})| \le  | f(\overline{x}, {\bf y})| $, and the induction hypothesis applies.
\item If $\max_{\overline{x}, f(\overline{x},{\bf y}), t}  \{ |\overline{x}| + T_1(g_a), |f(\overline{x},{\bf y})| +T_2(g_a), |\overline{y_t}| + T_{t+2}(g_a), c_{g_a} \} = |\overline{y_t}| + T_{t+2}(g_a)$ for some $t$: the result applies immediately by structural induction on $g_a$. 
\item If $\max_{\overline{x}, f(\overline{x},{\bf y}), t}  \{ |\overline{x}| + T_1(g_a), |f(\overline{x},{\bf y})| +T_2(g_a), |\overline{y_t}| + T_{t+2}(g_a), c_{g_a} \} = c_{g_a}$, the result holds immediately.
\item The base case $f(\epsilon, {\bf y})$ is immediate. 
\end{enumerate}

\item Let $f$ be defined by primitive recursion on notations, and assume now that the second of the two corresponding cases holds. Let $f(a.\overline{x},{\bf y}) = g_a(\overline{x}, f(\overline{x},{\bf y}),{\bf y})$, for $a\in\{0,1\}$.
Since the first case does not hold, $T_2(g_0)= 1$ or $T_2(g_1) = 1$: assume that $T_2(g_0)= 1$ (the other case being symmetric). Assume also that, $T_1(g_0) \le 0$ and $T_1(g_1) \le 0$, and for all $t\ge 2$, $T_t(g_0)=T_t(g_1)=T_{t-2}(h)=-\infty$. Then, we set $T_1(f) = \max \{0, c_h\}$. We apply the tropical interpretation on $g$, and prove by induction on the length of $a.\overline{x}$ that $|f(a.\overline{x},{\bf y})| \le |a.\overline{x}| + \max \{ c_{g_1}, c_{g_2}, c_h \}$.
\begin{enumerate}
\item If $\max_{t\ge 2} \{ |\overline{x}| + T_1(g_a), |f(\overline{x},{\bf y})| + T_2(g_a), |\overline{y_t}|+ T_t(g_a),  c_{g_a} \} = |\overline{x}| + T_1(g_a)$. Since $T_1(g_a)\le0$ and $T_1(f) \ge 0$, $| f(a.\overline{x}, {\bf y})| \le |\overline{x}| \le T_1(f) + |\overline{x}|$, and the result holds.
\item If $\max_{t\ge 2} \{ |\overline{x}| + T_1(g_a), |f(\overline{x},{\bf y})| + T_2(g_a),  |\overline{y_t}|+ T_t(g_a), c_{g_a} \}  = |f(\overline{x})| + T_2(g_a)$. Since $T_2(g_a)\le1$,  $| f(a.\overline{x})| \le 1 +  | f(\overline{x})| $, and the induction hypothesis allows to conclude.
\item If $\max_{t\ge 2} \{ |\overline{x}| + T_1(g_a), |f(\overline{x},{\bf y})| + T_2(g_a),  |\overline{y_t}|+ T_t(g_a),c_{g_a} \} = c_{g_a}$, the result holds immediately.
\item The case $\max_{t\ge 2} \{ |\overline{x}| + T_1(g_a), |f(\overline{x},{\bf y})| + T_2(g_a),  |\overline{y_t}|+ T_t(g_a),c_{g_a} \} = |\overline{y_t}|+ T_t(g_a)$ is impossible since $T_t(g_a) = -\infty$ for $t\ge2$.
\item The base case $f(\epsilon,{\bf y})$ is immediate. 
\end{enumerate}

\item Let now assume $f$ is define by primitive recursion on values. Then, the only possible case is similar to the first case of primitive recursion on notation. 
\end{enumerate}

The proof by induction above emphasizes the critical difference between recursion on notation and recursion on values: the second case of the safe recursion on notations correspond to the linear, non-size increasing scanning of the input, as in, for instance, 
$$ f(a.\overline{x}) = s_a(f(\overline{x})).$$ 
This, of course, is only possible in recursion on notation, where the height of the recursive calls stack is precisely the length of the scanned input. Recursion on values fails to perform this linear scanning, since, for a given recursive argument $\overline{x}$, the number of recursive calls is then exponential in $|\overline{x}|$. 
\end{proof}

Proposition~\ref{prop:trop-inter} proves that the tropical tiering of a function yields actually a tropical polynomial interpretation for the function symbols: The right hand side of the Lemma inequality is indeed a tropical interpretation. 
Moreover, this interpretation is directly given by the syntax.

Furthermore, the proof also highlights why we use mutual recursion schemes instead of more simple, non-mutual ones: non-size increasing discipline forbids the use of multiplicative constants in the size of the strings. So, in order to capture a computational space of size $k.\log(n)$, we need to use $k$ binary strings of length $\log(n)$, defined by mutual recursion. 


\begin{corollary}\label{cor:trop-inter}
L-primitive pointer recursive functions are computable in logarithmic space.
\end{corollary}
\begin{proof}

Proposition~\ref{prop:trop-inter} ensures that the size of all binary strings is logarithmically bounded. A structural induction on the definition of $f$ yields the result. The only critical case is that of a recursive construct. When evaluating a recursive construct, one needs simply to store all non-recursive arguments (the $\overline{y_i}$'s) in a shared memory, keep a shared counter for keeping track of the recursive argument $\overline{x}$, and use a simple ${\tt while}$ loop to compute successively all intermediate recursive calls leading to $f(\overline{x},{\bf y})$. All these shared values have logarithmic  size. The induction hypothesis ensures then that, at each step in the ${\tt while}$ loop, all computations take logarithmic space. The two other cases, composition and basic functions, are straightforward.
\end{proof}

In the following section, we prove the converse: logarithmic space functions can be computed by a sub-algebra of the L-primitive pointer recursive functions.

\subsection{Tropical Recursion}\label{subsec:trop}

In this section we restrict our tropical tiering approach to only four possible tier values: $1$, $0$, $-1$ and $-\infty$. The rules for tiering are adapted accordingly. More importantly, the use of only four tier values allows us to denote these tropics directly in the syntax, in an approach similar to that of Bellantoni and Cook. Let us take as separator symbol the following $\tv$ symbol, with leftmost variables having the highest tier. As with safe recursive functions, we allow the use of a  high tier variable in a low tier position, as in, for instance,  
$$f(\overline{x} \tv \overline{y} \tv \overline{z} \tv \overline{t}  ) = g ( \tv \overline{y} \tv \overline{x}, \overline{z}\tv \overline{t} ).$$
Our tropical recursive functions are then as follows. 

\begin{paragraph}{Basic tropical pointer functions.}
Basic tropical pointer functions are the following.
$$
\begin{array}{rclrcl}
{\tt hd}(\tv \tv \tv a.\overline{x}) & = & a & {\tt tl}(\tv \tv a.\overline{x} \tv  ) &  = & \overline{x}\\
{\tt hd}(\tv \tv  \tv  \varepsilon) & = & \varepsilon & {\tt tl}(\tv\tv \varepsilon \tv ) &  = & \varepsilon\\
{\tt s}_0(\overline{x}\tv\tv\tv) &=& 0.\overline{x} & {\tt s}_1(\overline{x}\tv\tv\tv ) &=& 1.\overline{x}\\
{\tt s}(\overline{x}\tv\tv\tv) &=& E(D(\overline{x})+1) \; \; \; \; \; \;\; \; \;\; \; \;\; \; \; & {\tt Read} ( \tv \tv \tv \overline{x}) &=& a \in \{ 0, 1 \}\\
\multicolumn{4}{r}{{\tt Pr}_i^n(\tv \overline{x_i} \tv \tv \overline{x_1}, \cdots, \overline{x_{i-1}}, \overline{x_{i+1}}, \cdots, \overline{x_n})} &=& \overline{x_i}\\
\end{array}$$
\end{paragraph}
\begin{paragraph}{Tropical composition.}
Define ${\bf t} = {\bf t}_1,{\bf t}_2,{\bf t}_3,{\bf t}_4$.
The tropical composition scheme is then
\begin{eqnarray*}
f({\bf x}\tv {\bf y} \tv {\bf z} \tv {\bf t} ) &=&g(h_1(\tv {\bf x}\tv {\bf y} \tv {\bf t} ) , \cdots, h_a(\tv {\bf x}\tv {\bf y} \tv {\bf t} ) \tv \\
&&h_{a+1}({\bf x}\tv {\bf y} \tv {\bf z}\tv {\bf t} ) , \cdots, h_b({\bf x}\tv {\bf y} \tv {\bf z}\tv {\bf t} ) \tv \\
&&h_{b+1}( {\bf y} \tv {\bf z} \tv \tv {\bf t}) , \cdots, h_c( {\bf y} \tv {\bf z} \tv \tv {\bf t})  \tv \\
&&h_{c+1}({\bf t}_1\tv {\bf t}_2 \tv {\bf t}_3 \tv {\bf t}_4) , \cdots, h_d({\bf t}_1\tv {\bf t}_2 \tv {\bf t}_3 \tv {\bf t}_4) ) \\
\end{eqnarray*}
\end{paragraph}
\begin{paragraph}{Tropical Recursion on Notations - case 1.}
\begin{eqnarray*}
f(  {\bf x}\tv \varepsilon,{\bf y} \tv {\bf z} \tv {\bf t}) &=& h({\bf x}\tv {\bf y} \tv {\bf z}\tv {\bf t})\\
f(  {\bf x}\tv {\tt s}_a(\overline{r}\tv\tv\tv),{\bf y} \tv {\bf z} \tv {\bf t}) &=& g_a (   {\bf x} \tv \overline{r}, f( {\bf x}\tv \overline{r},{\bf y} \tv {\bf z} \tv {\bf t}), {\bf y} \tv {\bf z}\tv {\bf t} )\\
\end{eqnarray*}
\end{paragraph}
\begin{paragraph}{Tropical Recursion on Notations - case 2.} (Linear scanning)
\begin{eqnarray*}
f(\tv \varepsilon \tv \tv {\bf t} ) &=& \varepsilon\\ 
f(\tv {\tt s}_a(\overline{r}\tv\tv\tv) \tv \tv {\bf t}) &=& g_a ( f(\tv \overline{r}\tv \tv {\bf t} ) \tv \overline{r} \tv \tv {\bf t} )
\end{eqnarray*}
\end{paragraph}
\begin{paragraph}{Tropical Recursion on Values.}
\begin{eqnarray*}
f(  {\bf x}\tv \varepsilon,{\bf y} \tv {\bf z} \tv {\bf t}) &=& h({\bf x}\tv {\bf y} \tv {\bf z}\tv {\bf t})\\
f(  {\bf x}\tv {\tt s}(\overline{r}\tv\tv\tv),{\bf y} \tv {\bf z} \tv {\bf t}) &=& g (   {\bf x} \tv \overline{r}, f( {\bf x}\tv \overline{r},{\bf y} \tv {\bf z} \tv {\bf t}), {\bf y} \tv {\bf z}\tv {\bf t} )\\
\end{eqnarray*}
\end{paragraph}

As above, we use the mutual version of these  recursion schemes, with the same tiering discipline. Note that, unlike previous characterizations of sub-polynomial complexity classes~\cite{DBLP:journals/cc/Bloch94,DBLP:journals/iandc/BonfanteKMO16,DBLP:journals/cc/Kuroda04}, our tropical composition and recursion schemes are only syntactical refinements of the usual composition and primitive recursion schemes - removing the syntactical sugar yields indeed the classical schemes. 

\begin{definition}{L-tropical functions}\\
The class of L-tropical functions is the closure of our basic tropical pointer functions, under tropical composition, tropical mutual recursion on notations, and tropical mutual recursion on values.
\end{definition}

The restriction of only four tier values suffices to capture the computational power of RATMs. More precisely, 

\begin{theorem}\label{thm:logspace}
The class of L-tropical functions is the class of functions computable in logarithmic space, with logarithmic size output. 
\end{theorem}
\begin{proof}
L-tropical functions are L- primitive pointer recursive functions with tropics $1$, $0$, $-1$ and $-\infty$. Following Corollary~\ref{cor:trop-inter}, they are computable in logspace. The converse follows from the simulation of a logarithmic space RATM. The simulation works as follows.
\begin{paragraph}{Encoding the machine configurations.}
Assume the machine $M$ works in space $k\lceil \log(n+1) \rceil$. A configuration of $M$ is then encoded by $2k + 3$ binary strings of length less than $\lceil \log(n+1) \rceil$:
\begin{enumerate}
\item one string, of constant length, encodes the machine state,
\item one string, of length $\lceil \log(n+1) \rceil$, encodes the pointer tape,
\item one string, of length $\lceil \log(n+1) \rceil$, encodes the head of the pointer tape. It contains ${\bf 0}$ symbols everywhere,  but on the position of the head (where it contains a ${\bf 1}$).
\item $k$ strings, of length $\lceil \log(n+1) \rceil$, encode the content of the work tape, and
\item $k$ strings, of length $\lceil \log(n+1) \rceil$, encode the position of the work tape head, with (as for the pointer tape)  ${\bf 0}$ everywhere but on the position of the head.
\end{enumerate}
\end{paragraph}
\begin{paragraph}{Reading and Updating a configuration.}
Linear scanning of the recursive argument in tropical recursion, corresponding to case 2 of the definition of tropical recursion on notations, is used to read and to update the encoding of the configuration. In order to do so, one  defines L-tropical functions for
\begin{enumerate}
\item  encoding boolean values {\tt true} and {\tt false}, boolean connectives,  and {\tt if then else}  constructs,
\item scanning an input string until a {\bf 1} is found, and computing the corresponding prefix sequence,
\item computing left and right extractions of sub-strings of a string, for a given prefix,
\item replacing exactly one bit in a binary string, whose position is given by a prefix of the string. 
\end{enumerate}
With all these simple bricks, and especially with the in-place one-bit replacement, one is then able to read a configuration, and to update it, with L-tropical functions. None of these L-tropical functions uses recursion on values.
\end{paragraph}
\begin{paragraph}{Computing the Transition map of the Machine.}
Given the functions above, the transition map ${\bf Next}$ of the machine is then computed by a simple L-tropical function of width $(2k+3)$: For a recursive argument $\overline{s}$ of size $\lceil \log(n+1) \rceil$, ${\bf Next}(\tv \overline{s}, {\bf c} \tv \tv)$ computes the configuration reached from ${\bf c}$ in one transition step. 
\end{paragraph}
\begin{paragraph}{Simulating the RATM.}
The simulation of the RATM is then obtained by iterating its transition map a suitable number of times. The time upper bound is here obtained by nesting $k$ tropical recursive functions on values: on an input of size $\lceil \log(n+1) \rceil$, the unfolding of these recursive calls takes time $n^k$. At each recursive step, this function needs to apply the transition map. The transition map having width $(2k+3)$, we use here a mutual recursion scheme, of width $(2k+3)$. Again, for a recursive argument $\overline{s}$ of size $\lceil \log(n+1) \rceil$, we define 
\begin{eqnarray*}
{\bf Step}_1 ( \tv \varepsilon, \overline{s}, {\bf c}\tv\tv) &=&  {\bf c}\\
{\bf Step}_1 ( \tv {\tt s}(\overline{t}\tv \tv\tv), \overline{s},{\bf c} \tv\tv ) &=& {\bf Next}(\tv \overline{s},\bf  {\tt Step}_1(\tv \overline{t}, \overline{n},{\bf c}\tv\tv)\tv\tv)\\
{\bf Step}_2 ( \tv \varepsilon, \overline{s}, {\bf c}\tv\tv) &=& {\bf c}\\
{\bf Step}_2 ( \tv {\tt s}(\overline{t}\tv \tv\tv), \overline{s}, {\bf c} \tv\tv ) &=& {\bf Step}_1(\tv \overline{s}, \overline{s}, {\bf Step}_2(\tv \overline{t}, \overline{s}, {\bf c}\tv\tv)\tv\tv)\\
&\vdots&\\
{\bf Step}_k ( \tv \varepsilon, \overline{s}, {\bf c}\tv\tv) &=& {\bf c}\\
{\bf Step}_k ( \tv {\tt s}(\overline{t}\tv \tv\tv), \overline{s}, {\bf c}\tv\tv ) &=& {\bf Step}_{k-1}(\tv \overline{s}, \overline{s}, {\bf Step}_k(\tv \overline{t}, \overline{s}, {\bf c}\tv\tv)\tv\tv).
\end{eqnarray*}
Replacing $\overline{s}$ by the {\tt Offset} in the above gives the correct bounds. 
\end{paragraph}
Finally, one simply needs to use simple L-tropical functions for computing the initial configuration, and reading the final configuration.
\end{proof}

\section{Logarithmic Space, Polylogarithmic Time}\label{sec:LP}

\begin{definition}{LP-tropical functions}\\
The class of LP-tropical functions is the closure of our basic tropical pointer functions, under tropical composition and tropical mutual recursion on notations.
\end{definition}

\begin{theorem}\label{thm:logspacepolylogtime}
The class of LP-tropical functions is the class of functions computable in logarithmic space, polylogarithmic time, with logarithmic size output. 
\end{theorem}
\begin{proof}
Mutual recursion on notations, with recursive arguments of logarithmic size, are computable in polylogarithmic time, following similar arguments as in the proof of Theorem~\ref{thm:logspace}. The converse follows from the simulation in the proof of Theorem~\ref{thm:logspace} above, where mutual recursion on values for the functions ${\tt Step}_i$ is replaced by mutual recursion on notations. 
\end{proof}


\section{Alternation}\label{sec:alt}

In this section we extend the approach of Leivant and Marion~\cite{DBLP:conf/csl/Marion94} to our setting. 
Let us define a similar tropical recursion on notations with substitutions. Note that the tropical tiering discipline prevents using substitutions in case 2 of the tropical recursion on notations. Substitutions are therefore only defined for case 1 of this recursion scheme.

\begin{paragraph}{Tropical Recursion with substitutions on Notations.} Given functions $h$, $g_0$, $g_1$, $k _1$ and $k_2$, 
$$
\begin{array}{rcl}
f(  {\bf x}\tv \varepsilon,\overline{u},{\bf y} \tv {\bf z} \tv {\bf t} ) &=& h({\bf x}\tv \overline{u}, {\bf y} \tv {\bf z} \tv {\bf t})\\
f(  {\bf x} \tv {\tt s}_a(\overline{r}\tv\tv\tv), {\bf y} \tv \overline{u},{\bf z} \tv {\bf t} ) &=&g_a ( {\bf x} \tv  \overline{r}, f( {\bf x}\tv \overline{r},k_1(\tv \overline{u}\tv\tv),{\bf y} \tv {\bf z}\tv {\bf t} ),\\
\multicolumn{3}{r}{  f( {\bf x}\tv \overline{r},k_2(\tv \overline{u}\tv\tv),{\bf y} \tv {\bf z}\tv {\bf t} ),   {\bf y} \tv {\bf z} \tv {\bf t})}.\\
\end{array}
$$

\end{paragraph}
\begin{paragraph}{Tropical Recursion with substitutions on Values.} Given functions $h$, $g$, $k _1$ and $k_2$, 
$$
\begin{array}{rcl}
f(  {\bf x}\tv \varepsilon,\overline{u},{\bf y} \tv {\bf z} \tv {\bf t} ) &=& h({\bf x}\tv \overline{u}, {\bf y} \tv {\bf z} \tv {\bf t})\\
f(  {\bf x} \tv {\tt s}(\overline{r}\tv\tv\tv), {\bf y} \tv \overline{u},{\bf z} \tv {\bf t} ) &=&g ( {\bf x} \tv  \overline{r}, f( {\bf x}\tv \overline{r},k_1(\tv \overline{u}\tv\tv),{\bf y} \tv {\bf z}\tv {\bf t} ),\\
\multicolumn{3}{r}{  f( {\bf x}\tv \overline{r},k_2(\tv \overline{u}\tv\tv),{\bf y} \tv {\bf z}\tv {\bf t} ),   {\bf y} \tv {\bf z} \tv {\bf t})}.\\
\end{array}
$$
\end{paragraph}
Again, as above, we assume these recursion schemes to be mutual.

\begin{definition}{P-tropical functions}\\
The class of P-tropical functions is the closure of our basic tropical pointer functions, under tropical composition, tropical recursion on notations and on values, and tropical recursion with substitutions on notations and on values.
\end{definition}

\begin{theorem}\label{thm:P}
The class of P-tropical functions with binary output is the class P.
\end{theorem}
\begin{proof}
The result follows from $Alogspace = P$~\cite{DBLP:journals/jacm/ChandraKS81}, and Theorem~\ref{thm:logspace}. Substitutions in the tropical recursion scheme on notations amounts to alternation. Restriction to decision classes instead of function classes comes from the use of alternating Turing machines, which compute only decision problems. 

Let us first see how to simulate a logspace alternating machine with P-tropical functions.
Recall the notations and functions of the proof of Theorem~\ref{thm:logspace}. Since we now need to simulate a non-deterministic, alternating machine, we assume without loss of generality that we now have two kinds of machine states:
\begin{itemize}
\item non-deterministic universal
\item non-deterministic existential
\end{itemize}
and that non-deterministic transitions have at most two branches. Therefore, we also assume that we have one predicate that determines the kind of a state in a configuration ${\bf c}$:
$ {\tt IsUniversal} (\tv\overline{s}, {\bf c}\tv\tv )$. This predicate is assumed to output {\tt false}  or {\tt true}.

We also assume that we have two transition maps, ${\tt Next}_0(\tv \overline{s}, {\bf c}\tv\tv)$, and ${\tt Next}_1(\tv \overline{s}, {\bf c}\tv\tv)$, for computing both branches of non-deterministic transitions. For deterministic transitions, we assume both branches are the same. Finally, we also assume we have a predicate ${\tt isPositive}(\tv \overline{s}, {\bf c}\tv \tv )$, which returns ${\tt true}$ if the configuration ${\bf c}$ is final and accepting, and ${\tt false}$ otherwise. 

We define now, with substitutions, the following:
$$
\begin{array}{rcl}
{\tt Accept} (\tv \varepsilon, \overline{s},{\bf c}\tv \tv ) &=& {\tt isPositive}(\tv\overline{s}, {\bf c} \tv \tv ) \\
{\tt Accept} (\tv {\tt s}(\overline{t}),\overline{s}, {\bf c} \tv\tv)) &=& {\tt match } \; {\tt IsUniversal} (\tv \overline{s}, {\bf c}\tv \tv  ) \; {\tt with}\\
\multicolumn{3}{l}{\;\; | {\tt true} -> {\tt AND}\;( \tv  {\tt Accept} (\tv\overline{t}, {\tt Next}_0(\tv \overline{s},{\bf c}\tv\tv),{\bf c}\tv\tv),}\\
\multicolumn{3}{r}{\;\;  {\tt Accept} (\tv\overline{t}, {\tt Next}_1(\tv\overline{s}, {\bf c}\tv\tv))\tv\tv)} \\
\multicolumn{3}{l}{\;\; | {\tt false} -> {\tt OR}\;(  \tv  {\tt Accept} (\tv\overline{t}, {\tt Next}_0(\tv \overline{s},{\bf c}\tv\tv)),{\bf c}\tv\tv),}\\
\multicolumn{3}{r}{\;\;  {\tt Accept} (\tv\overline{t}, {\tt Next}_1(\tv \overline{s},{\bf c}\tv\tv))\tv\tv)}.
\end{array}
$$
Then, for $\overline{t}$ and $\overline{s}$ large enough, and an initial configuration ${\bf c}$, ${\tt Accept}(\tv \overline{t},\overline{s},{\bf c}\tv\tv)$ outputs the result of the computation of the machine. Finally, nesting up to $k$ layers of such recursion on values schemes allows, as in the proof of Theorem~\ref{thm:logspace}, to simulate a polynomial computation time.


The other direction is pretty straightforward: For any instance of a recursion scheme with substitutions, for any given values $\overline{r}$, $\overline{u}$, ${\bf x}$,  ${\bf y}$ and ${\bf z}$, each bit of \\
$g_a ( {\bf x} \tv  \overline{r}, f( {\bf x}\tv \overline{r},k_1(\tv \overline{u}\tv\tv),{\bf y} \tv {\bf z}\tv {\bf t} ), f( {\bf x}\tv \overline{r},k_2(\tv \overline{u}\tv\tv),{\bf y} \tv {\bf z}\tv {\bf t} ),   {\bf y} \tv {\bf z} \tv {\bf t})$
is a boolean function of the bits of $f( {\bf x}\tv \overline{r},k_1(\tv \overline{u}\tv\tv),{\bf y} \tv {\bf z}\tv {\bf t} )$ and $f( {\bf x}\tv \overline{r},k_2(\tv \overline{u}\tv\tv),{\bf y} \tv {\bf z}\tv {\bf t} )$. Hence, it can be computed by an alternating procedure. The space bound follows from the bound on the size of the strings, provided by the tiering discipline. 
\end{proof}

\begin{definition}{NC-tropical functions}\\
The class of NC-tropical functions is the closure of our basic pointer tropical functions, under tropical composition, tropical recursion on notations and tropical recursion with substitutions on notations.
\end{definition}

\begin{theorem}\label{thm:NC}
The class of NC-tropical functions with binary output is NC.
\end{theorem}
\begin{proof}
The result follows from $A(logspace, polylogtime) = NC$ ~\cite{DBLP:journals/jcss/Ruzzo81}, and Theorem~\ref{thm:logspace}. Substitutions in the tropical recursion scheme on notations amounts to alternation. The proof is similar to that of Theorem~\ref{thm:P}, where additionally, 
\begin{itemize}
\item The time bound on the computation of the machine  needs only to be polylogarithmic, instead of polynomial. As in Theorem~\ref{thm:logspacepolylogtime}, tropical recursion on  notations suffices to obtain this bound, and tropical recursion on values is no longer needed. 
\item For the other direction, any bit of\\
 $g_a ( {\bf x} \tv  \overline{r}, f( {\bf x}\tv \overline{r},k_1(\tv \overline{u}\tv\tv),{\bf y} \tv {\bf z}\tv {\bf t} ),$ $f( {\bf x}\tv \overline{r},k_2(\tv \overline{u}\tv\tv),{\bf y} \tv {\bf z}\tv {\bf t} ),   {\bf y} \tv {\bf z} \tv {\bf t})$ is again a boolean function of the bits of $f( {\bf x}\tv \overline{r},k_1(\tv \overline{u}\tv\tv),{\bf y} \tv {\bf z}\tv {\bf t} )$ and $f( {\bf x}\tv \overline{r},k_2(\tv \overline{u}\tv\tv),{\bf y} \tv {\bf z}\tv {\bf t} )$. Here, this boolean function can be computed by a boolean circuit of polylogarithmic depth, hence, by an alternating procedure in polylogarihtmic time. The arguments behind this remark are the same as the ones in the proof of $A(logspace, polylogtime) = NC$.
\end{itemize}
\end{proof}

\section{Concluding Remarks}

Theorems~\ref{thm:logspace}, ~\ref{thm:logspacepolylogtime}, ~\ref{thm:P}, and~\ref{thm:NC} rely on mutual recursive schemes. As stated above, we use these mutual schemes to express a space computation of size $k \log(n)$ for any constant $k$, with binary strings of length at most $\log(n)+c$. If we were to use only non-mutual recursion schemes, we would need to have longer binary strings. This can be achieved by taking as input to our functions, not simply the {\tt Offset}, but some larger string $\#^k({\tt Offset})$, where $\#^k$ is a function that appends $k$ copies of its argument. 

It also remains to be checked wether one can refine Theorem~\ref{thm:NC} to provide characterizations of the classes NC$^i$ as in ~\cite{DBLP:journals/tcs/LeivantM00}. A first step in this direction is to define a recursion rank, accounting for the nesting of recursion schemes: then, check wether NC-tropical functions of rank $i$  are computable in NC$^i$. Conversely, check also whether the simulation of  Theorem~\ref{thm:logspace} induces a fixed overhead, and wether NC$^i$ can be encoded by NC-tropical functions of rank $i+c$ for some constant $c$ small enough. 

Finally, note that we characterize logarithmic space functions with logarithmically long output (Theorem ~\ref{thm:logspacepolylogtime}), and NC functions with one-bit output (Theorem ~\ref{thm:NC}). As usual, polynomially long outputs for these classes can be retrieved via a pointer access: it suffices to parameterize these functions with an additional, logarithmically long input, denoting the output bit one wants to compute. In order to retrieve functions with polynomially long output, this approach could also be added to the syntax, with a {\tt Write} construct similar to our {\tt Read} construct, for writing the output. 

\bibliographystyle{eptcs}
\bibliography{./PointerRec}

\end{document}